\documentclass[11pt]{article}

\usepackage{a4wide}
\usepackage{amssymb}
\usepackage{amsmath}
\usepackage{amsthm}
\usepackage[dvips]{graphics,color}
\usepackage{epsfig}
\usepackage{natbib}
\def\var{\text{Var}}
\def\cov{\text{Cov}}

\newtheorem{thm}{Theorem}
\newtheorem{corl}{Corollary}

\begin{document}

\title{Dynamic modeling of mean-reverting spreads for statistical arbitrage}

\author{
K. Triantafyllopoulos \footnote{Department of Probability and
Statistics, Hicks Building, University of Sheffield, Sheffield S3
7RH, UK, email: {\tt k.triantafyllopoulos@sheffield.ac.uk} } \and G.
Montana \footnote{Department of Mathematics, Statistics Section,
Imperial College London, London SW7 2AZ, UK, email: {\tt
g.montana@imperial.ac.uk } } }

\date{\today}

\maketitle

%\tableofcontents
%\newpage

\begin{abstract}

Statistical arbitrage strategies, such as pairs trading and its
generalizations, rely on the construction of mean-reverting spreads
enjoying a certain degree of predictability. Gaussian linear
state-space processes have recently been proposed as a model for
such spreads under the assumption that the observed process is a
noisy realization of some hidden states. Real-time estimation of the unobserved spread process
can reveal temporary market inefficiencies which can then be exploited
to generate excess returns. Building on previous work, we embrace the state-space framework for modeling spread processes and extend this methodology along three different directions.
First, we introduce time-dependency in the model parameters, which
allows for quick adaptation to changes in the data generating
process. Second, we provide an on-line estimation algorithm that can
be constantly run in real-time. Being computationally fast, the
algorithm is particularly suitable for building aggressive trading
strategies based on high-frequency data and may be used as a
monitoring device for mean-reversion. Finally, our framework
naturally provides informative uncertainty measures of all the
estimated parameters. Experimental results based on Monte Carlo
simulations and historical equity data are discussed, including a
co-integration relationship involving two exchange-traded funds.

\textit{Keywords:} mean reversion, pairs trading, state-space
models, time-varying autoregressive processes, dynamic regression,
statistical arbitrage.

\end{abstract}

\section{Introduction}\label{introduction}

A time series is known to exhibit mean reversion when, over a
certain period of time, is ``reverting'' to a constant mean. In
recent years, the notion of mean reversion has received a
considerable amount of attention in the financial literature. For
instance, there has been increasing interest in studying the
long-run properties of stock prices, with particular attention being
paid to investigate whether stock prices can be characterized as
random walks or mean reverting processes. If a price time series
evolves as a random walk, then any shock is permanent and there is
no tendency for the price level to return to a constant mean over
time; moreover, in the long run, the volatility of the process is
expected to grow without bound, and the time series cannot be
predicted based on historical observations. On the other hand, if a
time series of stock prices follows a mean reverting process,
investors may be able to forecast future returns by using past
information. Since the seminal work of \cite{Fama1988a} and
\cite{Poterba1988}, who first documented mean-reversion in stock
market returns during a long time horizon, several studies have been
carried out to detect mean reversion in several markets (e.g.
\cite{Chaudhuri2003}) and many asset classes (e.g. \cite{Deaton1992,
Jorion1996}).

Since future observations of a mean-reverting time series can
potentially be forecasted using historical data, a number of studies
have also examined the implications of mean reversion on portfolio
allocation and asset management; see \cite{Barberis2000} and
\cite{Carcano2005} for recent works. Active asset allocation
strategies based on mean-reverting portfolios, which generally fall
under the umbrella of \emph{statistical arbitrage}, have been
utilized by investment banks and hedge funds, with varying degree of
success, for several years. Possibly the simplest of such strategies
consists of a portfolio of only two assets, as in \emph{pairs
trading}. This trading approach consists in going long a certain
asset while shorting another asset in such a way that the resulting
portfolio has no net exposure to broad market moves. In this sense,
the strategy is often described as \emph{market neutral}. Entire
monographs have been written to illustrate how pairs trading works,
how it can be implemented in real settings, and how its performance
has evolved in recent years (see, for instance,
\cite{Vidyamurthy2004} and \cite{Pole2007}). The underlying
assumption of pairs trading is that two financial instruments with
similar characteristics must be priced more or less the same.
Accordingly, the first step consists in finding two financial
instruments whose prices, in the long term, are expected to be tied
together by some common stochastic trend. What this implies is that,
although the two time series of prices may not necessarily move in
the same direction at all times, their spread (for instance, the
simple price difference) will fluctuate around an equilibrium level.
Since the spread quantifies the degree of mispricing of one asset
relative to the other one, these strategies are also refereed to as
\emph{relative-value}. If a common stochastic trend indeed exists
between the two chosen assets, any temporary deviation from the
assumed mean or equilibrium level is likely to correct itself over
time. The predictability of this portfolio can then be exploited to
generate excess returns: a trader, or an algorithmic trading system,
would open a position every time a substantially large deviation
from the equilibrium level is detected and would close the
position when the spread has reverted back to the its mean. This
simple concept can be extended in several ways, for instance by
replacing one of the two assets with an artificial one (e.g. a
linear combination of asset prices), with the purpose of exploiting
the same notions of relative-value pricing and mean-reversion,
although in different ways; some relevant work along these lines has
been documented, among others, by \cite{Montana2008} and
\cite{Montana2008c}, who describe statistical arbitrage strategies
involving futures contracts and exchange-traded funds (ETFs),
respectively. One aspect that has not been fully investigated in the studies above
is how to explicitly model the resulting observed spread. A
stochastic model describing how the spread evolves over time is
highly desirable because it allows the analyst to precisely
characterize and monitor some of its salient properties, such as
mean-reversion. Moreover, improved trading rules may be built around
specific properties of the adopted spread process.

Recently,
\cite{Elliott2005} suggested that Gaussian linear state-space
processes may be suitable for modeling mean-reverting spreads
arising in pairs trading, and described how such models can yield
statistical arbitrage strategies. Their main observation is that the
observed process should be seen as a noisy realization of an
underlying hidden process describing the true spread, which may
capture the true market conditions; thus, a comparison of the
estimated unobserved spread process with the observed one may lead
to the discovery of temporary market inefficiencies. Based on the
additional assumption that the model parameters do not vary over
short periods of time, \cite{Elliott2005} suggested to use the EM
algorithm, an iterative procedure for maximum likelihood estimation,
for tracking the hidden process and estimating the other unknown
model parameters. To make the exposition self-contained, we briefly
review their model in Section \ref{mean_reversion}, and state under
what conditions the stochastic process is mean-reverting.

In this paper we build upon the model by \cite{Elliott2005} and
extend their methodology in a number of ways. First, in Section
\ref{newmodel}, we introduce time-dependency in the model
parameters. The main advantage of this formulation is a gain in
flexibility, as the model is able to adapt quickly to changes in the
data generating process; Section \ref{motivations} further motivates
our formulation and discusses its potential advantages. In Section
\ref{model} we derive new conditions that need to be satisfied for
a model with time-varying parameters to be mean-reverting. In
Section \ref{bayesian} we describe a Bayesian framework for
parameter estimation which then leads to a recursive parameter
estimation procedure suitable for real-time applications. The final algorithm
is detailed in Section \ref{algorithm}; an analysis and discussion
on the convergence properties of the
algorithm as well as practical suggestions on how
to specify the initial values and prior distributions are
provided. Unlike the EM algorithm, our estimation procedure also
produces uncertainty measures without any additional computational
costs. With a view on
statistical arbitrage, in Section \ref{trading} we add a note
discussing how pairs trading may be implemented using the spread
models proposed in this work and enumerate other important issues
involved in realistic implementations, together with some pointers
to the relevant literature. However, an empirical evalutation of trading strategies
is beyond the scope of this work. For further discussions on statistical
arbitrage approaches based on mean-reverting spreads and many
illustrative numerical examples the reader is referred to \cite{Pole2007}.

In Section \ref{results}, based on a battery of Monte Carlo
simulations, we demonstrate that posterior means estimated on-line
by our Bayesian algorithm recovers the true model parameters and can
be particularly advantageous when the analysts wishes to track
sudden changes in the mean-level of the spread and its
mean-reverting behavior. For instance, real-time monitoring may be
used to derive stop-loss rules in algorithmic trading. Two examples
involving real historical data are given in Section \ref{equity},
where the cointegrating relationship between a pair of stocks and a
pairs of ETFs are discussed. Final remarks are found in Section
\ref{discussion} and the proofs of arguments in Sections \ref{model}
and \ref{converge} can be found in the appendix.

\section{Time-invariant state-space models for mean-reverting spreads} \label{mean_reversion}

Throughout the paper we will assume that the trader has identified
two candidate financial instruments whose prices are observed at
discrete time points $t=1,2,\ldots$ and are denoted by $p^{(j)}_t$,
with $j=1,2$. At any given time $t$, let $y_t$ denote the price
spread, defined as
$$
y_t=\alpha + p^{(i)}_t - \beta p^{(j)}_t
$$
for some parameters $\alpha$ and $\beta$ which are usually estimated by ordinary
least squares (OLS) methods using historical data. It seems common practice to
select the order of $i$ and $j$ such that $y_t$ yields the largest
$\beta$ and the resulting spread captures as much information as possible about the
(linear) co-movement of the two assets. In Section \ref{equity} we briefly
mention how a penalized OLS model may be used for recursive estimation of a time-varying $\beta$. More generally, the observed spread $y_t$ may also be obtained in different ways or may represent
the return process of an initial price spread. One
of the two component processes $\{p_t^{(j)}\}$ may even be artificially built using a
linear combination of a basket of assets. For our purposes, the only requirement is
that the process $\{y_t\}$ is assumed to be mean-reverting.

Furthermore, following \cite{Elliott2005}, we assume that the observed spread $y_t$ is a
noisy realization of a true but unobserved spread or \emph{state} $x_t$. The state process $\{x_t\}$ is defined such that
\begin{equation}\label{mrp1}
x_t - x_{t-1} = a-bx_{t-1}+\varepsilon_t
\end{equation}
where $0<b<2$, $a$ is an unrestricted real number and $x_1$ is the
initial state. The restriction $0<b<2$ is imposed, because otherwise $\{x_t\}$
is non-stationary and thus mean reversion has probability zero to occur. The innovation series $\{\varepsilon_{t}\}$ is taken
to be an i.i.d. Gaussian process with zero mean and variance $C^2$,
and $\varepsilon_{t+1}$ is assumed to be uncorrelated of $x_t$, for
$t=1,2,\ldots$. Conditions for the state process to be
mean-reverting are established using standard arguments, as follows.
First, rewrite (\ref{mrp1}) as
$$
x_t=a+(1-b)x_{t-1}+\varepsilon_t
$$
Expanding on this, we obtain
$$
x_t=(1-b)^{t-1}x_1+a\sum_{i=0}^{t-2}(1-b)^i+\sum_{i=0}^{t-2}(1-b)^i\varepsilon_{t-i}
$$
Then, taking expectations and variances,
$$
E(x_t)= (1-b)^{t-1}\left\{E(x_1)-\frac{a}{b}\right\}+\frac{a}{b}
$$
and
$$
\var(x_t)=(1-b)^{2(t-1)} \left\{ \var(x_1)-
\frac{1}{1-(1-b)^2}\right\} +\frac{C^2}{1-(1-b)^2}
$$

It is observed that, when $|1-b|<1$, and regardless of $a$,
$\lim_{t\rightarrow\infty}(1-b)^{t-1}=0$ and therefore
$\lim_{t\rightarrow\infty}E(x_t)=a/b$. Therefore, in the long run,
the state process fluctuates around its mean level $a/b$. Otherwise,
when $|1-b|\geq 1$, $(1-b)^{t-1}$ is unbounded and hence $E(x_t)$ is
unbounded too. Analogously, when $|1-b|<1$ and regardless of $a$,
the variance $\var(x_t)$ converges to $C^2/\{1-(1-b)^2\}$. Conversely,
if $|1-b|\geq 1$, the variance of $x_t$ is unbounded with geometric
speed. It is concluded that the hidden process $\{x_t\}$ is mean
reverting when $1-b$ lies inside the unit circle. Adopting the
notation of \cite{Elliott2005}, we define $A=a$ and $B=1-b$, so that
the process $x_t$ can be rewritten as
\begin{equation}
x_t=A+Bx_{t-1}+\epsilon_t \label{ss1}
\end{equation}
Without loss of generality, we postulate that $\{y_t\}$ is
a noisy version of $\{x_t\}$ generated as
\begin{equation}
y_t=x_t+\omega_t \label{ss2}
\end{equation}
where $\{\omega_t\}$ is Gaussian white noise with variance $D^2$ and
$\omega_t$ is uncorrelated of $x_t$, for $t=1,2,\ldots$. From
(\ref{ss2}), it also follows that $\{y_t\}$ is a mean-reverting
process.

Note that, together with an initial distribution of the state $x_1$,
equations (\ref{ss1}) and (\ref{ss2}) define a Gaussian linear
state-space model with parameters $A,B,C,D$. State-space models were
originally developed by control engineers \citep{KRE60} and are
useful tools for expressing dynamic systems involving unobserved
state variables. The reader is also referred to \cite{Harvey1989} and \cite{WH97}
for book-length expositions.

\section{Time-varying dynamic models and on-line estimation} \label{newmodel}

\subsection{Preliminaries} \label{motivations}

The linear Gaussian state-space model described by equations
(\ref{ss1}) and (\ref{ss2}) contains the unknown parameters
$A,B,C$ and $D$ which need to be estimated using historical data.
When the parameters are known, the Kalman filter provides a
recursive procedure for estimating the state process $x_t$
\citep{KRE60}. Full derivations of the Kalman filter and lucid
explanations in a Bayesian framework can be found in
\cite{Meinhold1983}. In practice, maximum likelihood estimation
(MLE) of the unknown parameters is required in order to fully
specify the model. MLE for state-space models can be routinely
carried out in a missing-data framework using the EM algorithm, as
first proposed in the 1980s by \cite{Shumway1982} and
\cite{Ghosh1989}; a detailed derivation can also be found in
\cite{Ghahramani1996}. In the context of pairs trading,
\cite{Elliott2005} reports some simulation and calibration studies
demonstrating that the EM algorithm provides a consistent and
robust estimation procedure for the model (\ref{ss1})-(\ref{ss2}),
and suggest that the finite-dimensional recursive filter described
in \cite{Elliott1999} may also be used for estimation (although no
results were provided).

\cite{Elliott2005} suggest to base model estimation on data points
belonging to a look-back window of size $N$. A full iterative
calibration procedure is then run till convergence every time a
new data point is observed and the window has been shifted
one-step ahead. This approach implicitly requires the analyst to
select a value of $N$ (the effective sample size) ensuring that, within each
time window, the model parameters do not vary. The selection of
$N$ may be difficult without a proper model selection procedure in
place to test the assumption that the model is locally
appropriate. For instance, although a small value of $N$ may
guarantee adequacy of the model, it could also lead to notable
biases in the parameter estimates. When $N$ is too large, a number
of factors such as special market events, persisting pricing
inefficiencies or structural price changes may invalidate the
modeling assumptions. Clearly, the question of how much
history to use to calibrate a model and the corresponding trading
strategy is a critical one.

From a practical point of view, repeating the EM algorithm several
times over different window sizes in search for an optimal window
size may be computationally expensive. Even performing a single
calibration run may not be fast enough to accommodate very
aggressive trading strategies in high-frequency domains, due to the
well-known slow convergence properties of the EM algorithm. More
notably, a vanilla application of this algorithm does not
automatically provide any measure of parameter uncertainty.
Although various methods and modifications have been proposed in
the statistical literature in this direction (see, for instance,
\cite{McLachlan1997}), the resulting methods usually introduce
further computational complexity.

In order to cope with these limitations, in this section we
present and discuss our three main contributions. Firstly, we
introduce more flexibility and release some of the modeling
assumptions by allowing the model parameters to vary over time; in
this way, both smooth and sudden changes in the data generating
process (such as those created by structural price changes and
unusual persistence of market inefficiencies) will be more easily
accommodated. Secondly, we propose a practical on-line estimation
procedure that, being non-iterative, can be run efficiently over
time, even at high sampling frequencies, and does not inflict the
burden of frequent re-calibration and window size selection.
Ideally, a model should be able to \emph{adapt} to changes in the
data generating mechanism with minimal user intervention, and
should be amenable to on-line monitoring so that the key
parameters characterizing the underlying mean-reverting property
can always be under continuous scrutiny. These features enable the
trader (or trading system) to take fast dynamic decisions.
Thirdly, as a result of the Bayesian framework proposed here
for recursive estimation, measures of uncertainty extracted from
the full posterior distribution can be routinely computed at no
extra cost. These measures can be very informative in quantifying
and assessing estimation errors, and can potentially be exploited
to derive more robust trading strategies; see Section
\ref{results} for some practical examples.

\subsection{The proposed model} \label{model}

In this section we initially propose a variation of the classic
state-space model used by \cite{Elliott2005} in which the
parameters are not assumed to be constant over time. This
modification will then force us to reconsider under which
conditions the spread process is mean-reverting.

First, let us rearrange the model (\ref{ss1}) and (\ref{ss2}) in
an autoregressive (AR) form. From \eqref{ss2}, note that
$x_t=y_t-\omega_t$. Then, from substitution in (\ref{ss1}) for
$t\geq 2$, we obtain
\begin{align} \label{ar}
y_t & = A+By_{t-1}+\epsilon_t
\end{align}
where $\epsilon_t=\omega_t-B\omega_{t-1}+\varepsilon_t$ is
distributed as a $N(0,\sigma^2)$, for $\sigma^2=D^2+B^2D^2+C^2$. The
above model is an AR model of order 1 with parameters $A$, $B$ and
$\sigma^2$.

We achieve time-dependence in the parameters of (\ref{ar}) by
replacing $A$ and $B$ with $A_t$ and $B_t$, respectively, and
postulating that both parameters evolve over time, according to
some weakly stationary process. Here we consider the case of $A_t$
and $B_t$ changing over time via AR models, but more general time
series may be considered. These choices lead to the specification
of a time-varying AR model of order 1, or TVAR(1). Accordingly,
the observed spread is described by the following law,
\begin{gather}
y_t =A_t+B_ty_{t-1}+\epsilon_t \label{tvar} \\  A_t=\phi_1
A_{t-1}+\nu_{1t}, \quad B_t=\phi_2B_{t-1}+\nu_{2t} \nonumber
\end{gather}
where $\phi_1$ and $\phi_2$ are the AR coefficients, usually being
assumed to lie inside the unit circle so that $A_t$ and $B_t$ may
be weakly stationary processes.

Setting $\theta_t=(A_t,B_t)'$ and
$F_t=(1,y_{t-1})'$, the model can be expressed in state space
form,
\begin{align}
y_t & =F_t'\theta_t+\epsilon_t \label{model1} \\  \theta_t &
=\Phi\theta_{t-1}+\nu_t\label{model1evol}
\end{align}
with $\Phi=\textrm{diag}(\phi_1,\phi_2)$ and error structure
governed by the observation error $\epsilon_t\sim N(0,\sigma^2)$ and the evolution error vector
$\nu_t=(\nu_{1t},\nu_{2t})'\sim N_2(0,\sigma^2V_t)$, where
$N_2(\cdot,\cdot)$ denotes the bivariate Gaussian distribution. It
is assumed that the innovation series $\{\epsilon_t\}$ and
$\{\nu_t\}$ are individually and mutually uncorrelated and they
are also uncorrelated of the initial state vector $\theta_1$, i.e.
$E(\epsilon_t\epsilon_s)=0;E(\nu_t\nu_s')=0;E(\epsilon_t\nu_u)=0;E(\epsilon_t\theta_1)=0;
E(\nu_t\theta_1')=0$, for any $t\neq s$, where $E(.)$ denotes
expectation and $\theta_1'$ denotes the row vector of $\theta_1$.

With the inclusion of a time component in the parameters $A$ and
$B$, we now need to revise the conditions under which the mean reversion
property holds true. The next
result gives sufficient conditions for the spread $\{y_t\}$ to be mean-reverting.
\begin{thm}\label{th1}
If $\{y_t\}$ is generated from model (\ref{model1})-(\ref{model1evol}), then, conditionally on a realized sequence
$B_1,\ldots,B_t$, $\{y_t\}$ is mean reverting if one
of the two conditions apply:
\begin{enumerate}
\item [(a)] $\phi_1=\phi_2=1$, $V_t=0$ and $|B_1|<1$; \item [(b)] $\phi_1$ and
$\phi_2$ lie inside the unit circle, $V_t$ is bounded and
$|B_t|<1$, for all $t\geq t_0$ and for some integer $t_0>0$.
\end{enumerate}
\end{thm}

Some comments are in order. First we note that if $A_t=A$ and
$B_t=B$ (this is achieved by setting $\phi_1=\phi_2=1$, and by
forcing the covariance matrix of $\nu_t$ to be zero for all $t$),
the condition $|B_1|=|B|<1$ of Theorem \ref{th1} reduces to the
known condition of mean reversion for the static AR model, as in
the previous section. In the dynamic case, when $A_t$ and/or $B_t$
change over time, the condition $|B_t|<1$ enables us to check mean
reversion in an on-line fashion. Following the approach of \cite{Elliott2005} for the AR model, we use model (\ref{model1})
in order to obtain estimates $\hat{B}_t$ of $B_t$ and then we check $|\hat{B}_t|<1$ in order to declare whether $\{y_t\}$ is mean reverting or not; in the following sections we detail the computations involved in the estimation of $B_t$. Structural changes in the level
of $y_t$ are accounted through estimates of $A_t$, but these do
not affect the mean reversion of $\{y_t\}$ as $A_t$ controls only
the level of $y_t$. For the case of $A_t=A$, this is explained in
some detail in \cite{Elliott2005} and for more information on
structural changes for cointegrated systems the reader is referred
to \cite{Johansen1988} and \citet[Chapter 6]{Lutkepohl2006}.
%With the assumed Gaussian AR model for $B_t$ we can easily
%calculate the limit of the probability that $B_t$ lies in the unit
%circle as
%\begin{equation}\label{prob1}
%\lim_{t\rightarrow\infty}
%P(|B_t|<1)=2\Phi\left(\frac{1}{\sqrt{\var(B_1)+\sum_{i=0}^\infty
%\phi_2^{2i}V_{2,t-i}}}\right)-1
%\end{equation}
%where $V_t=(V_{ij,t})_{i,j=1,2}$ and $\Phi(.)$ denotes the
%cumulative distribution function of $N(0,1)$. The proof of this
%result follows immediately as $B_t\sim
%N(\phi_2^{t-1}E(B_1),\var(B_1)+\sum_{i=0}^{t-2}\phi_2^{2i}V_{2,t-i})$
%converges in probability to $B\sim N(0,\var(B_1)+\sum_{i=0}^\infty
%\phi_2^{2i}V_{2,t-i})$ and of course the series $\sum_{i=0}^\infty
%\phi_2^{2i}V_{2,t-i}$ is convergent since $V_{2t}$ is bounded for
%all $t>0$. In the case of $V_t=V$, we obtain a simplification of
%(\ref{prob1}) as
%$$
%\lim_{t\rightarrow\infty}P(|B_t|<1)=2\Phi\left(\sqrt{\frac{1-\phi_2^2}{V_2}}\right)-1
%$$
The following result is a useful corollary of Theorem \ref{th1}.

\begin{corl}\label{cor1}
If $\{y_t\}$ is generated from model (\ref{model1})-(\ref{model1evol}) with
$\phi_1=1$, $|\phi_2|<1$, $V_{1t}=V_{12,t}=0$, then $\{y_t\}$ is
mean reverting if $|B_t|<1$, for all $t\geq t_0$, for
some $t_0>0$, where $V_t=(V_{ij,t})$.
\end{corl}
The proof of this corollary follows by combining the proofs of (a)
and (b) of Theorem \ref{th1} (see the appendix). Corollary
\ref{cor1} gives an important case, in which $A_t=A$, for all $t$ as
in \cite{Elliott2005}, but $B_t$ changes according to a weakly
stationary AR model. This can be used when it is expected that $A_t$
will be approximately constant and benefit may be gained by reducing
the tuning of the four parameters $\phi_1,\phi_2,\delta_1,\delta_2$
to tuning of two parameters $\phi_2,\delta_2$. For a further
discussion on this topic see Sections \ref{parameters} and
\ref{results}.

In this paper we propose (\ref{tvar}) as a flexible time-varying model for the observed spread. However,
more general time-varying autoregressive models may be used. Consider that $y_t$ is
generated from a time-varying AR model of order $d$, i.e.
\begin{equation}\label{bigmodel1}
y_t=A_t+\sum_{i=1}^d B_{it} y_{t-i}+\epsilon_t, \quad t\geq d+1
\end{equation}
and the time-varying AR parameters $A_t$ and $B_{it}$ follow first
order AR models, as
$$
A_t=\phi_1 A_{t-1}+\nu_{1t}, \quad B_{it} = \phi_{i+1}
B_{i,t-1}+\nu_{i+1,t}
$$
where $d$ is the order or lag of the autoregression, the innovations
$\epsilon_t$ and $\nu_{it}$ are individually and mutually
uncorrelated and they are uncorrelated with the initial states $A_d$
and $B_{id}$. Certain Gaussian distributions may be assumed on
$\epsilon_t$ and $\nu_{it}$ and on the states $A_t$ and $B_{it}$.
It is readily seen that this model can be casted in state space
form (\ref{model1}) with $F_t=(1,y_{t-1},\ldots,y_{t-d})'$,
$\theta_t=(A_t,B_{1t},\ldots,B_{dt})'$ and
$\Phi=\textrm{diag}(\phi_1,\ldots,\phi_{d+1})$ (the diagonal matrix
with diagonal elements $\phi_1,\ldots,\phi_{d+1}$). It is clear that
model (\ref{tvar}) is a special case of model (\ref{bigmodel1}) with
$d=1$. When the general model is adopted, the conditions of mean reversion of
$\{y_t\}$ of Theorem \ref{th1} need to be revised, as follows. For $\phi_i$
$(i=1,\ldots,d+1)$ being inside the unit circle, for $t>t_0$, all
(time-dependent) solutions of the autoregressive polynomial
$\psi(x)=1-\sum_{i=1}^d B_{it}x^i$ must lie outside the unit circle.
This effectively means that after some $t_0$, $\{y_t\}$ is locally
stationary \citep{Dahlhaus97}. For the remainder of this paper, we consider
the situation of $d=1$, i.e. model (\ref{tvar}), as this is a simple and parsimonious model.

\subsection{A Bayesian framework} \label{bayesian}

We adopt a Bayesian formulation that, within
the realm of conjugate analysis, allows us to derive fast recursive
estimation procedures and naturally compute measures of uncertainty.
The analysis we propose in this section has roots in the work of
\cite{West99,Prado02}, and \cite{trianta07}. Initially, we assume
that, given the observational variance $\sigma^2$, the initial state
$\theta_1$ follows a bivariate Gaussian distribution with mean
vector $m_1$ and covariance matrix $\sigma^2P_1$. Also, we place an
inverted gamma density prior with parameters $n_1/2$ and $d_1/2$ on
$\sigma^2$. In summary, the prior structure is specified as follows
\begin{equation}
\theta_1|\sigma^2 \sim N_2(m_1,\sigma^2 P_1) \quad \textrm{and}
\quad \sigma^2\sim IG(n_1/2,d_1/2),\label{priors1}
\end{equation}
where $m_1,P_1,n_1,d_1$ are assumed known; we comment on their
specification in Section \ref{parameters}. Note that,
unconditionally of $\sigma^2$, the initial state $\theta_1$ follows
a Student $t$ distribution.

With these priors in place, the posterior distribution of
$\theta_t|\sigma^2$ and the predictive distribution of
$y_t|\sigma^2$ are routinely obtained by the Kalman filter. We
elaborate more on this as follows. First, assume that at time $t-1$
the posteriors are given by $\theta_{t-1}|\sigma^2,y^{t-1}\sim
N_2(m_{t-1},\sigma^2P_{t-1})$ and $\sigma^2|y^{t-1}\sim
IG(n_{t-1}/2,d_{t-1}/2)$, for some $m_{t-1}$, $P_{t-1}$, $n_{t-1}$
and $d_{t-1}$. Here the notation $y^t$ means that all data points observed up to time $t$ are included.
Then, writing the likelihood function (or evidence) for
an observation $y_t$ as $p(y_t|\theta_t,\sigma^2)$, an application of the Bayes theorem gives
$$
p(\theta_t|\sigma^2,y^t) = \frac{ p(y_t|\theta_t,\sigma^2)
p(\theta_t|\sigma^2,y^{t-1}) }{ p(y_t|\sigma^2,y^{t-1}) },
$$
It follows that the posterior density of $\theta_t|\sigma^2$ is Gaussian, and specifically
$$
\theta_t|\sigma^2,y^t \sim N_2(m_t,\sigma^2P_t).$$
The recurrence equations for updating $m_t$ and $P_t$ are provided in Section \ref{algorithm}. The probability density
$p(y_t|\sigma^2,y^{t-1})$ refers to the one-step ahead forecast density,
which is obtained from the prior $p(\theta_t|\sigma^2,y^{t-1})$ as
$y_t|\sigma^2,y^{t-1}\sim N(f_t,\sigma^2Q_t)$. Again, see Section \ref{algorithm} below for recursive equations needed to update $f_t$ and $Q_t$.

The posterior distribution of $\sigma^2$ is also obtained by an
application of the Bayes theorem,
$$
p(\sigma^2|y^t)=\frac{ p(y_t|\sigma^2,y^{t-1}) p(\sigma^2|y^{t-1})
}{ p(y_t|y^{t-1}) }.
$$
which gives an inverted gamma density $\sigma^2|y^t\sim
IG(n_t/2,d_t/2)$, depending on parameters $n_t$
and $d_t$. Here $y_t|y^{t-1}$ follows a $t$ distribution with
$n_{t-1}$ degrees of freedom $y_t|y^{t-1} \sim
t(n_{t-1},f_t,Q_tS_{t-1})$, with $S_{t-1}=d_{t-1}/n_{t-1}$.

From the density $p(\theta_t|\sigma^2,y^t)$, the posterior
distribution of $\theta_t$, unconditionally of $\sigma^2$, is easily
obtained by integrating $\sigma^2$ out; it It then follows that
$\theta_t|y^t\sim t_2(n_t,m_t,P_tS_t)$. From this the $(1-\gamma)\%$
marginal confidence interval of $B_t$ is
$$
m_{2t}\pm t_{\gamma/2}\sqrt{P_{22,t}S_t}
$$
where $m_t=(m_{1t},m_{2t})'$, $P_t=(P_{ij,t})_{ij=1,2}$ and
$t_{\gamma}$ denotes the $100\gamma\%$ quantile of the standard $t$
distribution with $n_t$ degrees of freedom. The $(1-\gamma)\%$
confidence interval for $x_{t+1}$ is
$$
f_{t+1}\pm t_{\gamma/2}\sqrt{F_{t+1}'R_{t+1}F_{t+1}S_t}
$$
and the $(1-\gamma)\%$ confidence interval for $y_{t+1}$ is
$$
f_{t+1}\pm t_{\gamma/2}\sqrt{Q_{t+1}S_t}
$$
where the recurrence relationships of $R_{t+1}$ and $Q_{t+1}$ are
given below.

Some references on related time series models are in order. From a frequentist perspective, time varying AR models have been discussed in \cite{Dahlhaus97, Francq01, Francq04} and \cite{Anderson05}.
Among other works, recursive estimation of time varying autoregressive processes in a
nonparametric setting is discussed in \cite{Moulines2005} and, for
non Gaussian processes, in \cite{Djuric2002}, using particle
filters. Standard Bayesian AR models have been developed since the
early 70's, see e.g. \cite{Zellner72,Monahan83,Kadiyala97} and \cite{Ni03}. Free software for model estimation is widely available\footnote{Time-varying AR models
are implemented in the computing language R (website: {\tt http://cran.r-project.org/}) via the contributed package {\tt timsac}.
S-plus, Fortran and Matlab routines for the implementation of these models can be downloaded from the website of Mike West ({\tt http://www.stat.duke.edu/research/software/west/tvar.html})}.

\section{On-line estimation} \label{algorithm}

\subsection{An adaptive and recursive algorithm using discount factors}

In this section we provide the updating equations needed to
compute the posterior densities of $\theta_t|y^t$ and of
$\sigma^2|y^t$ at each time step. Starting at time $t=1$ with a
quadruple of initial values ($m_1,P_1,n_1,d_1$), the calibration
algorithm then proceeds as follows:
\begin{gather}
R_t  =\Phi P_{t-1} \Phi+V_t, \quad Q_t  =F_t'R_tF_t+1, \quad e_t
 =y_t-F_t'\Phi m_{t-1} \nonumber \\ K_t  =R_tF_t/Q_t, \quad
m_t  =\Phi m_{t-1}+K_te_t, \quad P_t =R_t-K_tK_t'Q_t \label{algo} \\
r_t =y_t-F_t'm_t , \quad n_t =n_{t-1}+1, \quad d_t =d_{t-1}+r_te_t,
\quad S_t  = \frac{d_t}{n_t} \nonumber
\end{gather}
For any $t=2,\ldots,T$, the above algorithm estimates the target
posterior quantities of interest; for instance, we can extract
posterior and predictive mean and variances, as well as relevant
quantiles and credible bounds of $\theta_t$ and $\sigma^2$. From
$\theta_t=(A_t,B_t)'$ and the posterior distribution of
$\theta_t|y^t$, we can extract the posterior distribution of
$B_t|y^t$. The condition for mean-reversion established in Theorem
\ref{th1} can be monitored recursively by extracting the posterior
mean of $B_t|y^t$, say $\hat B_t$, and assessing whether $|\hat
B_t|$ is strictly less than one. Credible bounds can also be
associated to the posterior mean in order to better assess the
possibility that the process is still mean-reverting -- see the
examples in Section \ref{results}.

The full specification of algorithm \eqref{algo} requires the
selection of a covariance matrix $V_t$, which is responsible for the
stochastic evolution of the signal $\theta_t$ and hence the
stochastic change of $A_t$ and $B_t$. Following \citet[Chapter 6]{WH97} we
advocate a practical and convenient analytical solution which allows
us to learn this variance component directly from the data in a
sequential way by means of two \emph{discount factors}, $\delta_1$
and $\delta_2$; this is referred to as \emph{component discounting}.
The idea is that by assuming $P_1$ and $V_t$ to be diagonal matrices
we can use the two discount factors to discount the precision of the
updating of the mean and the variance of $\theta_t$ as we move from
time $t-1$ to $t$. In other words we use $\delta_1$ and $\delta_2$
to specify the covariance matrix $V_t$ as
$$
V_t=\left(\begin{array}{cc} \delta_1^{-1}(1-\delta_1) \phi_1^2
p_{11,t-1} & 0 \\ 0 & \delta_2^{-1}(1-\delta_2) \phi_2^2 p_{22,t-1}
\end{array} \right)
$$
where $P_t=(p_{ij})_{i,j=1,2}$. This implies that
$R_t=\textrm{diag}(\phi_1^2 p_{11,t-1}/\delta_1,\phi_2^2
p_{22,t-1}/\delta_2)$ and thus, as we move from $t-1$ to $t$, the
prior variance of $A_t$ is increased by a factor of $1/\delta_1$ and
of $B_t$ by a factor of $1/\delta_2$. Of course if
$\delta_1=\delta_2=1$, then $V_t=0$ and in this case $\theta_t$
carries no stochastic evolution. If we allow $\delta_1=1$ and
$\delta_2<1$, then only $B_t$ has stochastic evolution over time.

\subsection{Model comparison and model
assessment}\label{diagnostics}

The performance of the estimation procedure of Sections
\ref{bayesian} and \ref{algorithm} can be formally evaluated using
model diagnostic and model comparison tools; see, for instance, \cite{Li03}
for a general exposition of time series diagnostics and
\cite{Harrison91} for diagnostics in state space models. In this section we
briefly discuss three diagnostic tools, namely the mean of the
squared standardized forecast errors (MSSE), the likelihood
function, and sequential Bayes factors.

From the Student $t$ distribution of $y_t|y^{t-1}$, i.e.
$y_t|y^{t-1}\sim t(n_{t-1},f_t,Q_tS_{t-1})$, we can define the
standardized one-step forecast errors (or standardized residuals)
as $u_t=Q_t^{-1}S_{t-1}^{-1}(y_t-f_t)$, so that $u_t|y^{t-1}\sim
t(n_{t-1},0,1)$ (the standard $t$ distribution with $n_{t-1}$
degrees of freedom). We can therefore construct diagnostics and
outlier detection tools based on the above $t$ distribution of
$u_t$. Writing $v_t=(1-2n_{t-1}^{-1})u_t$ we have
$E(v_t^2|y^{t-1})=1$ and so the MSSE is defined as
$(T-1)^{-1}\sum_{t=2}^T v_t^2$, which if the model fit is good,
should be close to 1.

From the Student $t$ distribution of $y_t|y^{t-1}$ the
log-likelihood function of $\phi_1,\phi_2,\delta_1,\delta_2$ based
on data $y^T=\{y_2,\ldots,y_T\}$ is
\begin{eqnarray*}
\ell(\phi_1,\phi_2,\delta_1,\delta_2;y^T) &=& \sum_{t=2}^T
p(y_t|y^{t-1}) \nonumber \\ &=& \sum_{t=2}^T \log
\frac{\Gamma(n_t/2)}{\sqrt{\pi n_{t-1}} \Gamma(n_{t-1}/2)} -
\frac{1}{2} \sum_{t=2}^T n_t \log \left\{
1+\frac{(y_t-f_t)^2}{n_{t-1}Q_tS_{t-1}}\right\}\label{LogL}
\end{eqnarray*}
where $\Gamma(.)$ denotes the gamma function. Model camparison can be carried out by using either one of the
following criteria: likelihood function, Akaike's information
criterion (AIC) and Bayesian information criterion (BIC)	. In particular, we can choose optimal values of
some or all of the hyperparameters
$\phi_1,\phi_2,\delta_1,\delta_2$ by maximizing $\ell(.)$. A
discussion on the specification of the hyperparameters of the
model can be found in Section \ref{parameters}.

For the application of the above diagnostic criteria, all data $y^T$ is needed to be available, or historical data can be used.
However, sometimes it is useful to construct sequential
diagnostics so that the model can be assessed and updated over
time in an adaptive way. Such diagnostics tools include sequential
likelihood ratios and sequential Bayes factors. Here we briefly
discuss the latter, the foundations of which are discussed in
detail in \citet[Chapter 11]{WH97}. Suppose that, given a sample
$y^T=\{y_1,\ldots,y_T\}$ we have two candidate models of the form
of (\ref{model1}) that is they have the same structural form, but
they may differ in the values of $\phi_1$, $\phi_2$, $\delta_1$
and $\delta_2$. Suppose that we denote the two models by
$\mathcal{M}_1$ and $\mathcal{M}_2$ and for $i=1,2$ we write
$\phi_{i1}$, $\phi_{i2}$, $\delta_{i1}$ and $\delta_{i2}$ to
indicate the dependence of model $\mathcal{M}_i$ in these
parameters. Then the Bayes factor of $\mathcal{M}_1$ versus
$\mathcal{M}_2$ is given by the ratio of their respective one-step
forecast densities, i.e.
$$
H_t=\frac{ p(y_t|y^{t-1},\mathcal{M}_1) }{
p(y_t|y^{t-1},\mathcal{M}_2) } = \left( \frac{
n_{t-1}Q_{2t}S_{2,t-1}+e_{2t}^2 }{ n_{t-1}Q_{1t}S_{1,t-1}+e_{1t}^2
} \right)^{n_t/2} \left(
\frac{Q_{1t}S_{1,t-1}}{Q_{2t}S_{2,t-1}}\right)^{n_t/2}
$$
where we have used that $y_t|y^{t-1},\mathcal{M}_i\sim
t(n_{t-1},f_{it},Q_{it}S_{i,t-1})$, with the quantities $f_{it}$,
$e_{it}$, $Q_{it}$, $S_{i,t-1}$ being appropriately indexed by
$i=1,2$. Given data $y^T$ one can either judge the performance of
the two models sequentially (by comparing $H_t$ to 1, for $2\leq
t\leq T$) and thus arriving to a sequential monitoring of the two
models, or use the entire data set $y^T$ to compare the models
globally, e.g. one can extract the mean or other features of the
empirical distribution of $\{H_t\}$.

\subsection{Convergence analysis}\label{converge}

Algorithm \eqref{algo} is quite similar to the celebrated Kalman
filter; conditional on $\sigma^2$, the algorithm exactly reduces to
the Kalman filter, but the full algorithm allows for the estimation
of $\sigma^2$ that results in the Student $t$ posterior distribution
for $\theta_t$. On the performance of the Kalman filter,
\cite{Elliott2005} state that the posterior covariance matrix of the
parameters converges to stable values and this has important
implications on the stability of the state process $\{x_t\}$.
Indeed, it is well known that if the parameters of a state space
model are constant, then the posterior covariance matrix of the
states converges to a stable value; see, for instance, \citet[p.
119]{Harvey1989} as well as \cite{Chan84, triantafyllopoulos07}.
However, the performance of the posterior covariance matrix $P_t$
when the components of the model are made time-dependent has not
been investigated; in our system this is conveyed via the
time-varying vector $F_t=(1,y_{t-1})'$. This aspect is important as
instability or divergence of $P_t$ could result in instability of
the estimation of $A_t$ and $B_t$ and hence of $x_t$. The next
result states that, in our system, $P_t$ converges to stable values
and we provide an explicit formula for the computation of the limit
of $P_t$.
\begin{thm}\label{th2}
Suppose that $\{y_t\}$ is generated from model (\ref{model1}). If
$\{y_t\}$ is mean reverting and if, for $j=1,2$, it is
$\delta_j<\phi_j^2$, then as $t\rightarrow\infty$ the limit $P$ of
the covariance matrix $P_t=\var(\theta_t|y^t)$ exists and it is
given by $P=\textrm{diag}(p_{11},p_{22})$, where
$$
p_{ii}=\left\{ \sum_{j=0}^\infty
\left(\frac{\delta_i}{\phi_i^2}\right)^j a_{i,t-j}\right\}^{-1}
$$
with $a_{1,t}=1$ and $a_{2,t}=y_{t-1}^2$.
\end{thm}
Some comments are in order. First we note that if
$\phi_1=\phi_2=1$ and $V_t=0$ (we have already seen that this
setting reduces the model to the time-invariant AR model
considered in \cite{Elliott2005}), then the condition
$\delta_j<\phi_j^2$ is satisfied for all values of $\delta_j$,
since $0<\delta_j<1$.

From the mean reversion assumption of $\{y_t\}$, if we write
$y_t\approx \mu$, where $\mu$ denotes the equilibrium mean of the
spread, then we can write the limit covariance matrix $P$ as
$$
P=\left(\begin{array}{cc} \phi_1^{-2} (\phi_1^2-\delta_1) & 0 \\ 0 &
\mu \phi_2^{-2} (\phi_2^2-\delta_2)\end{array}\right)
$$
In the important special case of $\phi_1=\delta_1=1$, for which
$A_t=A$ is time-invariant, we can easily see that
$$
P=\left(\begin{array}{cc} p_{11,1} & 0 \\ 0 & \left\{
\sum_{j=0}^\infty \left(\frac{\delta_2}{\phi_2^2}\right)^j
y_{t-j}^2\right\}^{-1}\end{array}\right)
$$
where $p_{11,1}$ is the prior variance $\var(A)$.

The convergence rate of the limit of Theorem \ref{th2} is geometric,
since after some appropriately large $t_L$, we can write $y_t\approx
\mu$, for all $t>t_L$ and the limit of $P$ depends on a geometric
series.

The above convergence results for $P_t$ are given conditional on the
variance $\sigma^2$. Given data up to time $t$, $\sigma^2$
has a posterior inverted gamma distribution $\sigma^2|y^t\sim
IG(n_t/2,d_t/2)$; hence, as the time index gets larger, the
variance of $\sigma^2$, which is given by
$$
\var(\sigma^2|y^t) = \frac{d_t^2}{(n_t-2)^2(n_t-4)^2} =
\frac{(n_1+t-1)^2S_t^2}{(n_1+t-3)^2(n_1+t-5)} \quad (t>5-n_1)
$$
converges to $0$. Therefore, as $t\rightarrow\infty$, $\sigma^2$
concentrates about its mode $S_t=d_t/n_t$ asymptotically
degenerating.

\subsection{Hyperparameter specification} \label{parameters}

The estimation algorithm \eqref{algo} relies upon the
specification of prior distributions and corresponding starting
values ($m_1,P_1,n_1,d_1$) and values of the model components
$(\phi_1,\phi_2,\delta_1,\delta_2)$, which are selected by the
user. In this brief section, considering
weakly informative priors, we provide some guidance on how to
choose these values. Of course, depending on the specific
application, other specifications may be preferred; for instance,
the analyst may want to include stronger prior beliefs regarding
the spread being traded, see e.g. \cite{Kadane1996}. Nevertheless,
it is important to note that, given a reasonable amount of data,
the sensitivity of the calibration procedure on these initial
specifications becomes negligible, especially over streaming data,
because the initial information is deflated over time. This
phenomenon is discussed in some detail in \cite{Ameen84} and in
\cite{trianta07}. Detailed studies on prior specification for the
estimation of AR models can be found in \cite{Kadiyala97,Ni03} and
in references therein.

The parameter $m_1$ is the prior mean of the hidden state, given the
observational variance, i.e. the mean of $\theta_1|\sigma^2$. A
common choice is to set $m_1$ equal to our prior expectation of
$(A_1,B_1)'$, which may be obtained from the availability of
historical data. In all examples of Section \ref{results} we have
used $m_1=(0,0)'$. This setting together with the vague prior $P_1$ that follows,
communicates a prior assumption of mean reversion, but with a large uncertainty placed
\emph{a priori} on $(A_1,B_1)'$.
The convergence results reported in Theorem
\ref{th2} above, guarantee that the choice of $m_1$ and $P_1$ are
not crucial for accurate estimation and forecasting. The covariance
matrix $P_1$ is chosen to be proportional to the $2\times 2$
identity matrix, i.e. $P_1=p_1I_2$. Here a large value of $p_1$
reflects a weakly informative or defuse prior specification, since
in this case the precision $P_1^{-1}$ gets close to zero. Finally,
values for $n_1$ and $d_1$ need to be provided. It can be noted
that, having placed an inverted gamma prior on $\sigma^2$, the
expected value of the observational variance is given by
$E(\sigma^2)=d_1/(n_1-2)$, for $n_1>2$. Based on this observation, a
sensible choice is to set $n_1=3$ and use the prior expectation of
$\sigma^2$ as a starting value $d_1$. Historical data may be used to
specify $d_1$, but in the examples of Section \ref{results} we have
simply used $d_1=1$.

Proceeding now with the specification of
$\phi_1,\phi_2,\delta_1,\delta_2$ we can optimize these parameters
by maximizing the log-likelihood function, given in Section
\ref{diagnostics}, under the condition that $\delta_i<\phi_i^2$ so
that Theorem \ref{th2} applies. Alternatively, according to
Corollary \ref{cor1} we can set $\phi_1=\delta_1=1$ and optimize
only $\phi_2$ and $\delta_2$. In Section \ref{simulations}, where we
present simulation studies, we use the latter, while in Section
\ref{equity}, where we analyze real data, we use the former (full
optimization of four parameters). We note that the likelihood function or Bayes factors
can be used to compare and optimize models using single discount factors $\delta_1=\delta_2$,
known as \emph{single discounting} \citep{WH97}, and models using two
different discount factors (component or multiple discounting).

\subsection{Pairs trading} \label{trading}

Under the assumption that the observed spread process involving two
tradable assets is mean-reverting, and that the model of Eqs.
\eqref{ss1}-\eqref{ss2} describes well its evolution at discrete
observational times $t=t_1,\ldots,t_N$, with $N$ sufficiently small,
a simple pairs trading strategy immediately follows
\citep{Elliott2005}. Let us assume that $\hat x_t$ denotes our best
estimate of the hidden state, which is obtained by calibrating the
model on data collected in the above data window.

At each time $t$, if the observed spread $y_t$ is strictly greater
than the true state $\hat x_t$, then a sensible decision would be to
take a long position in this portfolio, with the intention of
closing this position at a later time, when the spread has reverted
back to its mean. Conversely, if $y_t < \hat x_t$, the trader may
decide to take a short position in the portfolio; this bet is
expected to be a profitable one as soon as the spread process
corrects itself again. Realistic implementations of this popular
strategy may ask for additional layers of sophistication which in
turn require the trader to face a few practical questions; some
examples are:

\begin{itemize}
 \item How can transaction costs be included in this simple model?  In other words,
 when is a trade expected to be profitable, so that an `entry' signal can be
 generated? For instance, a long position could be initiated when $y_t - \hat x_t > z_t$,
 where $z_t$ is a threshold that guarantees a profitable trade, after costs. The question
 then becomes, how should $z_t$ be calibrated? For instance, \cite{Vidyamurthy2004} suggests
 a re-sampling procedure and provides some general guidance. There may exist several other
 alternative ways in which one could define entry points, perhaps based on empirical modeling of
 the extreme values of the $y_t-x_t$ process. Theoretical results on zero-crossing rates for
 autoregressive processes, as in \cite{Cheng1995}, may also be explored. For aggressive
 strategies that execute a trade at each single time tick, \cite{Montana2008a} forecasts
 the one-step ahead expected spread using dynamic regression methods, whereas \cite{Montana2008b}
 embrace the principle of ``learning with experts'' to deal with the uncertaincty involved in future movements on the spread.

\item Analogously to the previous issue, how should an `exit' signal be generated? And shall
a trade be closed at an exit point, or simply reversed so that a long position becomes short, and viceversa?

 \item What stop-loss mechanism can be implemented to make sure that the assumptions on which
 the strategy relies are still satisfied? Surely, if the spread process is no longer believed to
 be mean-reverting, a stop-loss signal should be quickly generated. As will appear clearer later
 (see, for instance, the examples of Section \ref{results}), our estimation procedure can be used
 to monitor mean-reversion sequentially and flag deviations from the acceptable behaviour of the
 spread process as soon as they occur. Related co-integration arguments may also be used, as in \cite{Lin2006}.

 \item How can suitable pairs of assets be chosen in the first place, especially when the
 universe of assets to search from is extremely large? Since arbitrage profits between two
 assets depend critically on the presence of a long-term equilibrium between them (see,
 for instance, \cite{Alexander2002a}), data mining methods built around co-integration
 techniques may be explored, as in \cite{dAspremont2008}. See also \cite{Vidyamurthy2004}
 and \cite{Pole2007} for alternative methods including simple correlation analysis, turning
 point analysis and latent factor models.
\end{itemize}

As a final note, we mention a technique that may be deployed in a
dynamic modeling setting, such as ours, to obtain the spread $y_t =
p^{(1)}_t - \beta_t p^{(2)}_t$ in  a recursive fashion. As noted
before, the regression coefficient is usually estimated on
historical data, but on-line procedures such as recursive least
squares may also be used. The assumption of a time-invariant
regression coefficient $\beta$ could also be released so as to allow
$\beta$ to change slightly over time; such a modification would
capture a time-varying co-integration relationship between the two
asset prices, where this extension deemed necessary. Assuming $T$
historical observations, a regression model with a time-varying
regression coefficient $\beta_t$ minimizes a cost function
\begin{equation} \label{original_cost}
C(\beta; \mu) = \sum_{t=1}^T \left\{p^{(1)}_t - \beta_t p^{(2)}_t  \right\}^2 + \mu
\sum_{t=1}^{T-1} (\beta_{t+1}-\beta_t)^2
\end{equation}
where $\mu \geq 0$ is a scalar determining how much penalization to place on temporal changes in the regression coefficient. When $\mu$ is very large, changes in the coefficient are penalized more heavily and, in the limit $\mu=\infty$, the usual OLS estimate is recovered. A solution to the optimization problem above was originally proposed by \cite{Kalaba1988}. Following their approach, called \emph{flexible least squares} (FLS), a recursive estimator for each $\beta_t$ can easily be derived as
\begin{equation}
\hat{\beta}_t =
\left[S_{t-1}+\{p_t^{(2)}\}^2\right]^{-1}\left\{s_{t-1}+p_t^{(1)}p_t^{(2)}\right\}
\label{eq:Betat_1}
\end{equation}
where we have defined the quantities
\begin{gather}
S_t=\mu \left[S_{t-1}+\mu I_p +\{p_t^{(2)}\}^2\right]^{-1} \left\{S_{t-1}+\{p_t^{(2)}\}^2\right\}\label{eq:St} \\
s_t = \mu \left[S_{t-1}+\mu I_p
+\{p_t^{(2)}\}^2\right]^{-1}\left\{s_{t-1}+p_t^{(1)}p_t^{(2)}\right\}\nonumber
\end{gather}
The recursions are initially started with some arbitrarily chosen
values $S_1$ and $s_1$. \cite{Montana2008} show a clear algebraic
connection between FLS and the Kalman filter and use this estimation
method to develop a dynamic statistical arbitrage strategy.

\section{Illustrations} \label{results}

\subsection{Simulated data}\label{simulations}

In this section we initially report on a Monte Carlo simulation
study demonstrating that the fast recursive algorithm \eqref{algo}
described in Section \ref{bayesian} accurately estimates the
parameters of the proposed model. We have simulated a large number
of time series under model \eqref{model1} using a range of values
for $A,B$ and $\sigma^2$. The true parameters are kept constant in
these initial simulations for simplicity, so they can be easily
compared with the estimated posterior means. We have found that convergence to the true
parameters $A$, $B$ and $\sigma^2$ is quickly achieved and the estimated
values of these parameters are not sensitive to the initial
parameters $\mu_1,P_1,n_1,d_1$ (results not shown).

We have also explored situations in which the parameters are
time-varying. First, we have considered the case of a sudden change
in the level of the spread; the time series fluctuates around an
equilibrium level till $t=1500$, and after that time it jumps to a
much higher equilibrium. Clearly for $1\leq t\leq 1499$ the process
is mean reverting, then at $t=1500$ it looses mean reversion, but it
retains it in the sub-period $1500\leq t\leq 3000$; of course the
process is not mean reverted for the entire period $1\leq t\leq
3000$. Figure \ref{fig:Ajumpest} shows how the posterior mean of
$|B_t|$ is tracked using two different values of the discount factor
$\delta_2$. Our focus is on monitoring $B_t$ because, as established
in Theorem \ref{th1}, this parameter is the ultimate object of
interest. As shown in Figure \ref{fig:Ajumpest}, the algorithm with
$\delta_2=1$ (which corresponds to a model with time-invariant
parameters) does not manage to capture the loss of mean-reversion
observed at time $t=1500$; in fact the algorithm gives the
misleading result of mean reversion throughout the time range. On
the contrary, when using a smaller discount factor (which
corresponds to a model with time-varying parameters), the algorithm
tracks the jump almost in real-time and communicates the result that
after $t=1500$ the process has locally regained mean reversion.

Furthermore, we have considered a more hypothetical scenario that
may be of practical interest: Figure \ref{fig:Bjump} corresponds to
a scenario where $B_t$ is piece-wise constant and undergoes a large
sudden jump at time $t=1500$. Again, the algorithm is able to track
well mean reversion locally, although the true parameter $B_t$ may
not be estimated very accurately.

\vspace*{1cm}
\begin{center}
FIGURES 1-2 AROUND HERE
\end{center}
\vspace*{1cm}

\subsection{Equity data} \label{equity}

In this section we apply our methods to spreads obtained from
historical equity data. Each spread is computed using the flexible
least squares (FLS) method with a very large $\mu$ parameter; this
is almost equivalent to ordinary least squares (OLS) regression but
allows for recursive estimation. As a simple validation exercise, we
also compare the findings obtained from our model to formal
cointegration tests which assume the availability of all data
points. The very first procedure for the estimation of cointegrating
regressions, based on OLS, was proposed by \cite{Engle1987}. Since
then several other procedures have been developed including the
maximum likelihood method of \cite{Johansen1988, Johansen1991} and
the fully modified OLS of \cite{Phillips1990}. \cite{Hargreaves1994}
lists eleven categories of procedures, and several more have been
added in more recent years. For our analysis we have considere only
three popular tests: Engle-Granger's ADF test \citep{Engle1987},
Phillips-Perron's PP test \citep{Perron1988} and Phillips-Ouliaris's
PO test \citep{Phillips1990a}.

The first data sets we present consists of daily share prices of two
companies: Exxon Mobil (XOM) and Southwest Airlines (LUV). We have
used all the available data for this pair of stocks, which spans a
period from March 23, 1980 to August 6, 2008. Figure \ref{spread1B}
reports the estimated posterior mean of $B_t$ and its confidence
band for the period March 23, 1980 to November 30, 2004. Clearly,
from March 23, 1980 till November 8, 2004 the posterior mean of
$|B_t|$ stays below one, which according to Theorem \ref{th1}
indicates mean-reversion of the spread time series. Figure
\ref{spread1} shows the observed spread time series as well as the
estimated hidden state process and its posterior confidence band for
this subperiod of the data. For the estimation of $(A_t,B_t)'$ we
have used $\phi_1=0.1$, $\phi_2=99839$, $\delta_1=0.992$ and
$\delta_2=0.995$ that maximize the log-likelihood function
(\ref{LogL}).

When using all historical data (1980-2008), all three standard
cointegration tests cannot reject the null hypothesis of unit roots
($p$-values: $0.246$, $0.219$ and $0.15$). This is in agreement with
the patterns captured by Figure \ref{spread1B}, which reveals that
after November 8, 2004, mean reversion is lost. However, when the
analysis is restricted to the period November 8, 2004, both the PP
and PO tests reject the null hypothesis of unit roots at a $5\%$
significance level ($p$-values: $0.013$ and $0.024$, respectively).
The ADF test, however, disagrees and does not reject the null
hypothesis of unit roots ($p$-value $0.139$). Thus, in this example,
only two out of three tests agree with the evidence provided by our
on-line monitoring device.

Our second example illustrates a co-integration relationship
existing between two ETFs operating in the commodity market. ETFs
are relatively new financial instruments that have exploded in
popularity over the last few years. They are securities that combine
elements of both index funds and stocks: like index funds, they are
pools of securities that track specific market indexes at a very low
cost; like stocks, they are traded on major stock exchanges and can
be bought and sold anytime during normal trading hours. We have
collected historical time series for the SPDR Gold Shares (GLD) and
Market Vectors Gold Miners (GDX) ETFs. GLD is an ETF that tries to
reflect the performance of the price of gold bullion, whereas GDX
tries to replicate as closely as possible, before fees and expenses,
the price and yield performance of the AMEX Gold Miners index. This
is achieved by investing in all of the securities which comprise the
index (in proportions given by their weighting in the index). This
analysis is based upon all the historical data available  for the
pair, which covers a shorted period compared to the previous
example, from May 23, 2006 until August 06, 2008. Figure
\ref{spread2} shows the observed spread process jointly with the
estimated hidden process and confidence bands, while Figure
\ref{spread2B} indicates that a co-integrating relationship between
the two ETFs does exist in the period from July 19, 2006 till 17
December, 2007. For this data set we have used $\phi_1=0.999$,
$\phi_2=99$, $\delta_1=0.95$ and $\delta_2=0.98$ that maximize the
log-likelihood function (\ref{LogL}).

When all the historical data is used, the ADF and the PP tests
indicate the presence of co-integration at a $5\%$ significance
level ($p$-values: $0.01$ and $0.01$, respectively) and only the PO
test suggest lack of co-integration, a result that also agrees with
the pattern reported in Figure \ref{spread2B}. Considering the
period July 19, 2006 till 17 December, 2007, for which our results
suggest mean reversion, we find that all three tests also suggest
co-integration ($p$-values: $0.0201$, $0.013$ and $0.012$). Further
formal comparisons and more detailed studies will be needed in order
to characterize some of the discrepancies; however, based on this
empirical evidence, our suggested time-varying model seems to
generally agree with most formal cointegration tests.

\bigskip

\vspace*{1cm}
\begin{center}
FIGURES 3-6 AROUND HERE
\end{center}
\vspace*{1cm}

\bigskip

\section{Conclusions} \label{discussion}

In this paper we have proposed a Bayesian time-varying
autoregressive model, expressed in time-space form, and an efficient
recursive algorithm based on forgetting or discount factors. The
procedure can be used for real-time estimation and tracking of the
underlying spread process and may be seen as a more efficient
alternative to standard iterative MLE procedures such as the EM
algorithm. Conditions for mean-reversion as well as the convergence
properties of the on-line estimation algorithm have been studied
analytically and discussed. The model seems particularly useful for
monitoring mean-reversion using financial data streams and as a
building block for statistical arbitrage strategies such as pairs
trading. Related algorithmic trading strategies that exploit
co-integration of financial instruments, for instance \emph{index
arbitrage} \citep{Sutcliffe2006} and \emph{enhanced index tracking}
\citep{Alexander2002a}, may also benefit from the methods proposed
here. Moreover, although the focus of this work has been on
applications in computational finance, we believe that the methods
described here are of broader interest and may appeal to other
users, within the management science community, who need to model
and monitor mean-reverting time series arising in different
application domains

There are several aspects of the suggested methodology that we would
like to explore further in future work. First, purely from an
empirical point of view, we would like to better understand how the
methodology relates to more formal statistical procedures for
testing the hypothesis of mean reversion based on finite sample
sizes. As already mentioned, since mean-reversion is closely linked
to second order stationarity, many efforts have been directed to
constructing unit root tests. These standard econometric procedures
may lack the power to reject the null hypothesis of a random walk,
and we feel that our method may at least complement them well.
Besides, some of the recently suggested procedures, such as the
bootstrap methods described by \cite{Li2003}, are too
computationally expensive to be of any use in the real settings and
applications that we have described. Another important aspect that
we plan to investigate is the question of how to learn the
discounting factors needed to specify the $V_t$ matrix in a more
adaptive fashion, so that they become self-tuning, rather than being
kept constant at all times. A number of techniques have been
successfully used for training adaptive artificial neural networks
and other time-varying stochastic processes using forgetting factors
\citep{Saad1999, Niedzwiecki2000} and there may be scope for
improvement along this direction.

\section*{Acknowledgements}
We would like to thank three anonymous referees for their helpful comments on an earlier draft of the paper.

\renewcommand{\theequation}{A-\arabic{equation}} % redefine the command that creates the equation no.
\setcounter{equation}{0}  % reset counter
\section*{Appendix}

\begin{proof}[Proof of Theorem \ref{th1}]
With $\phi_1=\phi_2=1$ and $V_t=0$, the state space model
(\ref{model1}) reduces to the AR model
$y_t=A+By_{t-1}+\epsilon_t$, where $A_t=A$ and $B_t=B$ and it is
trivial to verify that $\{y_t\}$ is mean reverting if
$|B_1|<1$, see also Section \ref{mean_reversion}. This completes (a).

Proceeding now to (b), from the AR model for $A_t$ we note that
$E(A_t)=0$. From (\ref{model1}) write $y_t$ recursively as
\begin{eqnarray*}
y_t&=&A_t+B_ty_{t-1}+\epsilon_t=A_t+B_tA_{t-1}+B_tB_{t-1}y_{t-2}+B_t\epsilon_{t-1}+\epsilon_t
= \cdots \\ &=& y_1\prod_{i=2}^t B_i +\sum_{j=0}^{t-3} \prod_{i=0}^j
B_{t-i}A_{t-j-1}+A_t+\sum_{j=0}^{t-3}\prod_{i=0}^jB_{t-i}\epsilon_{t-j-1}+\epsilon_t
\end{eqnarray*}
We write $A^t=(A_1,\ldots,A_t)$ and $B^t=(B_1,\ldots,B_t)$, for
$t=1,\ldots,T$. Since $\{\epsilon_t\}$ is white noise, we have
\begin{equation}
E(y_t|B^t)=y_1\prod_{i=2}^tB_i \label{eq1:proof}
\end{equation}
This is a convergent series if $|B_t|<1$, for all
$t>t_0$, for some positive integer $t_0$. To see this first write
$x_t^{(1)}=\prod_{i=2}^tB_i$, which is a decreasing series as
$|x_{t+1}^{(1)}/x_t^{(1)}|=|B_{t+1}|<1$. Also $\{x_t^{(1)}\}$ is
bounded as $|x_t^{(1)}|=\prod_{i=2}^t|B_i|<1$ and so
$\{x_t^{(1)}\}$ is convergent.

For the variance of $y_t$ we have
\begin{eqnarray}
\var(y_t|B^t) &=& \var(A_t) + \sum_{j=0}^{t-3} \prod_{i=0}^j
B_{t-i}^2 \var(A_{t-j-1}) + \sum_{j=0}^{t-3} \prod_{i=0}^j B_{t-i}^2
\var(\epsilon_{t-j-1}) \nonumber \\ && + \var(\epsilon_t) +
\sum_{j=0}^{t-3} \prod_{i=0}^j B_{t-i} \cov (A_t,A_{t-j-1})
\nonumber \\ &\leq & \sigma^2+\frac{\sigma^2 V_{11}}{1-\phi_1^2} +
\left(\frac{\sigma^2 V_{11}}{1-\phi_1^2}+\sigma^2\right)
\sum_{j=0}^{t-3} \prod_{i=0}^j B_{t-i}^2  + \frac{\sigma^2
V_{11}}{1-\phi_1^2} \sum_{j=0}^{t-3} \phi_1^{j+1} \prod_{i=0}^j
B_{t-i} \nonumber
\end{eqnarray}
where it is used that
$$
\var(A_t)\leq\frac{\sigma^2V_{11}}{1-\phi_1^2}  \quad \textrm{and}
\quad \cov(A_t,A_{t-j-1})\leq\frac{\sigma^2V_{11}}{1-\phi_1^2}
$$
for $V_{11,t} \leq V_{11}$, since from the hypothesis $V_t$ is
bounded, and so there exists some $V_{11}>0$ so that $V_{11,t}\leq
V_{11}$.

Now we show that the series $x_t^{(2)}=\sum_{j=0}^{t-3}
\prod_{i=0}^j B_{t-i}^2$ and $x_t^{(3)}=\sum_{j=0}^{t-3}
\phi_1^{j+1} \prod_{i=0}^j B_{t-i}$ are both convergent. For the
former series we note that given $|B_t|<1$, we can find some $B$
so that $|B_t|<|B|<1$, from which it follows that
$$
|x_t^{(2)}|\leq \sum_{j=0}^{t-3}\prod_{i=0}^j|B_{t-i}|\leq
\sum_{j=0}^{t-3}\prod_{i=0}^j|B|=\sum_{j=0}^{t-3}|B|^{j+1}
$$
which is proportional to a geometric series that converges for
$|B|<1$ and since $x_t^{(2)}$ is a positive series, it follows
that $\{x_t^{(2)}\}$ is convergent.

For the series $x_t^{(3)}$, we follow an analogous argument, i.e.
for $B$ satisfying $|B_t|<|B|<1$ we obtain
$$
|x_t^{(3)}|\leq \sum_{j=0}^{t-3} |\phi_1B|^{j+1}
$$
which shows that $x_t^{(3)}$ is convergent as $\sum_{j=0}^{t-3}
|\phi_1B|^{j+1}$ is a geometric series with $|\phi_1B|<1$ and
$x_t^{(3)}$ is a positive series.

With these convergence results in place, the convergence of
$\var(y_t|B^t)$ is obvious. Given, $B^t$, we have shown that the
mean and the variance of $\{y_t\}$ are convergent and so $\{y_t\}$
is mean reverting.
\end{proof}

\begin{proof}[Proof of Theorem \ref{th2}]
From the diagonal structure of
$P_t=\textrm{diag}(p_{11,t},p_{22,t})$ and the updating of $P_t$ as
in the calibration algorithm (\ref{algo}) we have
\begin{eqnarray*}
p_{ii,t} &=&\frac{\phi_i^2p_{ii,t-1}}{\delta_i} -
\frac{a_{it}\phi_i^4p_{ii,t-1}^2\delta_i^{-2}}{a_{it}\phi_i^2p_{ii,t-1}
\delta_i^{-1} +1} = \frac{\phi_i^2p_{ii,t-1}}{\delta_i}  \left(
1-\frac{a_{it}\phi^2p_{ii,t-1}}{\delta_i+a_{it}\phi_i^2p_{ii,t-1}} \right) \\
&=& \frac{\phi_i^2p_{ii,t-1}}{\delta_i+a_{it}\phi_i^2p_{ii,t-1}}
\end{eqnarray*}
We can clearly see that $p_{ii,t}>0$, for all $t$ and so we have
\begin{equation}\label{liminv}
\frac{1}{p_{ii,t}} = \frac{\delta_i}{\phi_i^2p_{ii,t-1}}+a_{it}  =
\frac{\delta_i^{t-1}}{\phi_i^{2t-2}p_{ii,1}}+\sum_{j=0}^{t-2}
\left(\frac{\delta_i}{\phi_i^2}\right)^ja_{i,t-j}
\end{equation}
Now since $\delta_i^j\phi_i^{-2j}a_{i,t-j}$ is a positive sequence and
since from the mean reversion of $\{y_t\}$ and the definition of
$a_{it}$, the above sequence is bounded above by the geometric sequence
$\delta_i^j\phi_i^{-2j} M$, where $M$ is an upper bound of
$\{y_t\}$, it follows immediately that $\sum_{j=0}^\infty
\delta_i^j\phi^{-2j}a_{i,t-j}<\infty$ and this proves that $p_{ii,t}$
converges to $\sum_{j=0}^\infty \delta_i^j\phi^{-2j}a_{i,t-j}$. The
proof is completed by inverting (\ref{liminv}), noting that
$p_{ii,t}>0$ and $\sum_{j=0}^\infty \delta_i^j\phi^{-2j}a_{i,t-j}>0$.
\end{proof}

\bibliographystyle{plainnat}
\bibliography{bibliography}

%%% Fig 1

\begin{figure}
\includegraphics[scale=0.75]{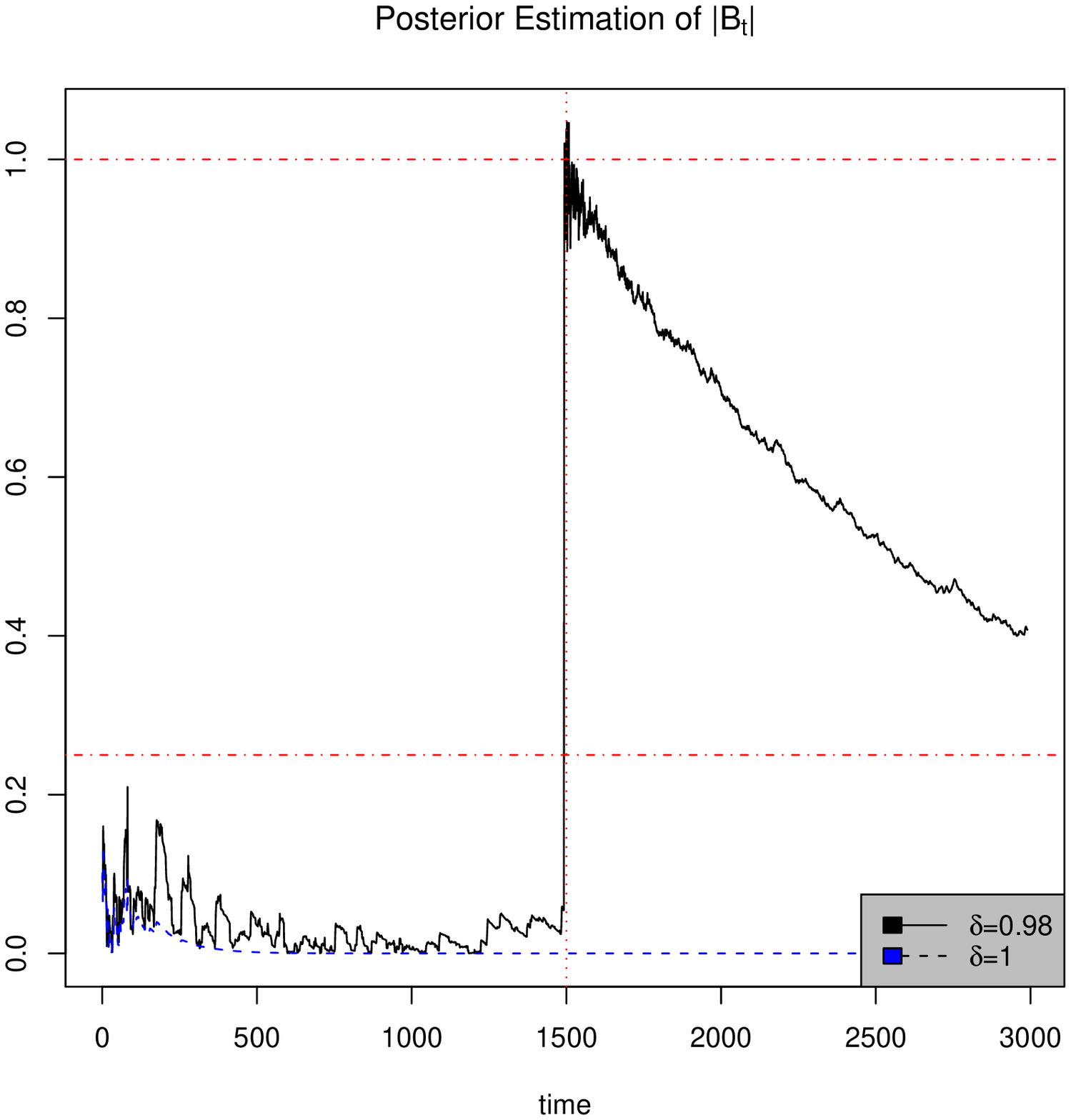}
\caption{Estimation of $|B_t|$ for the simulated spread with a jump
at $t=1500$. We have chosen a prior $P_1=1000I_2$. A value of
$\delta_1=\delta_2=\delta=1$, which corresponds to the adoption of a
time-invariant model, fails to capture mean-reversion following
immediately after the change of equilibrium ar time $t=1501$.
However, forgetting factors $\delta_1=1$ and $\delta_2=\delta=0.98$
tracks the abrupt change in mean level and and following quick
restoration of mean-reversion.} \label{fig:Ajumpest}
\end{figure}

\newpage

%%% Fig 2

\begin{figure}
\includegraphics[scale=0.75]{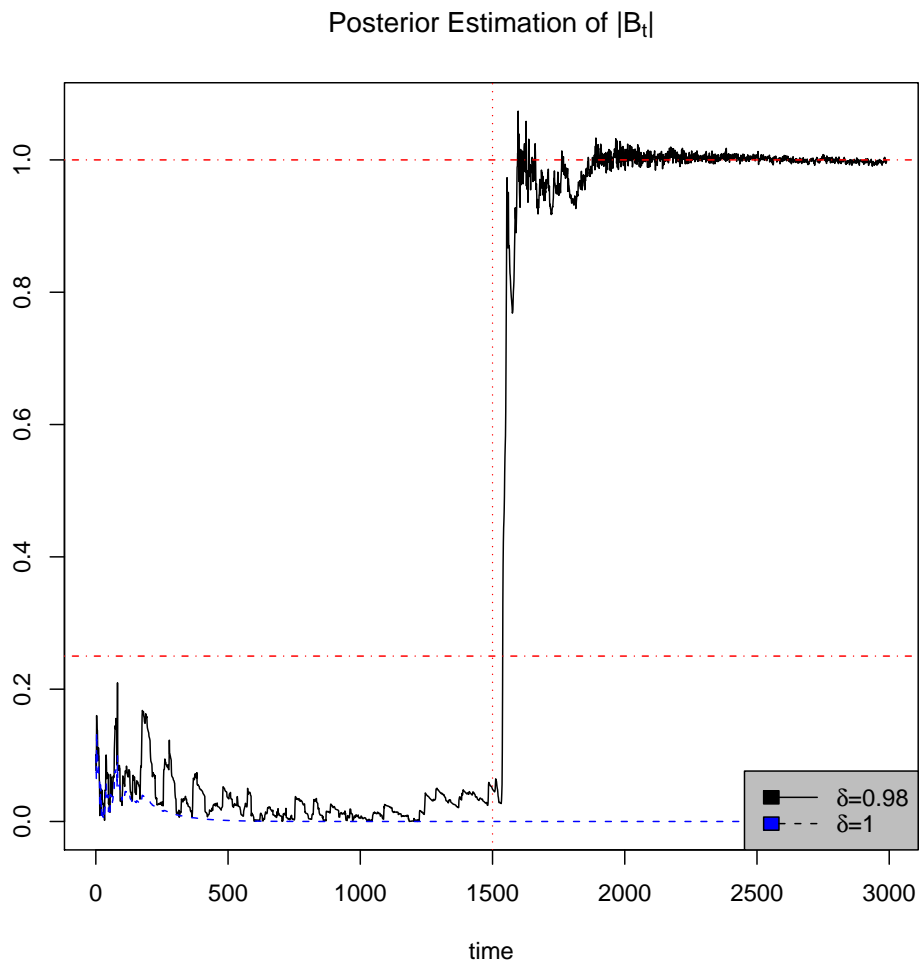}
\caption{Estimation of abruptly varying $B_t$; shown is the
posterior mean of $|B_t|$. The real parameters are $A=0.2$,
$B_t=0.25$ and $\sigma^2=1$, for $1\leq t\leq 1500$; $A=0.2$,
$B_t=1$ and $\sigma^2=1$, for $1501\leq t\leq 3000$. We have chosen
a prior $P_1=1000I_2$ and $\delta_1=1$. Two selected values of
$\delta_2=\delta$ are used, $\delta=1$ and $\delta=0.98$.}
\label{fig:Bjump}
\end{figure}

\newpage

\begin{figure}
\includegraphics[scale=0.75]{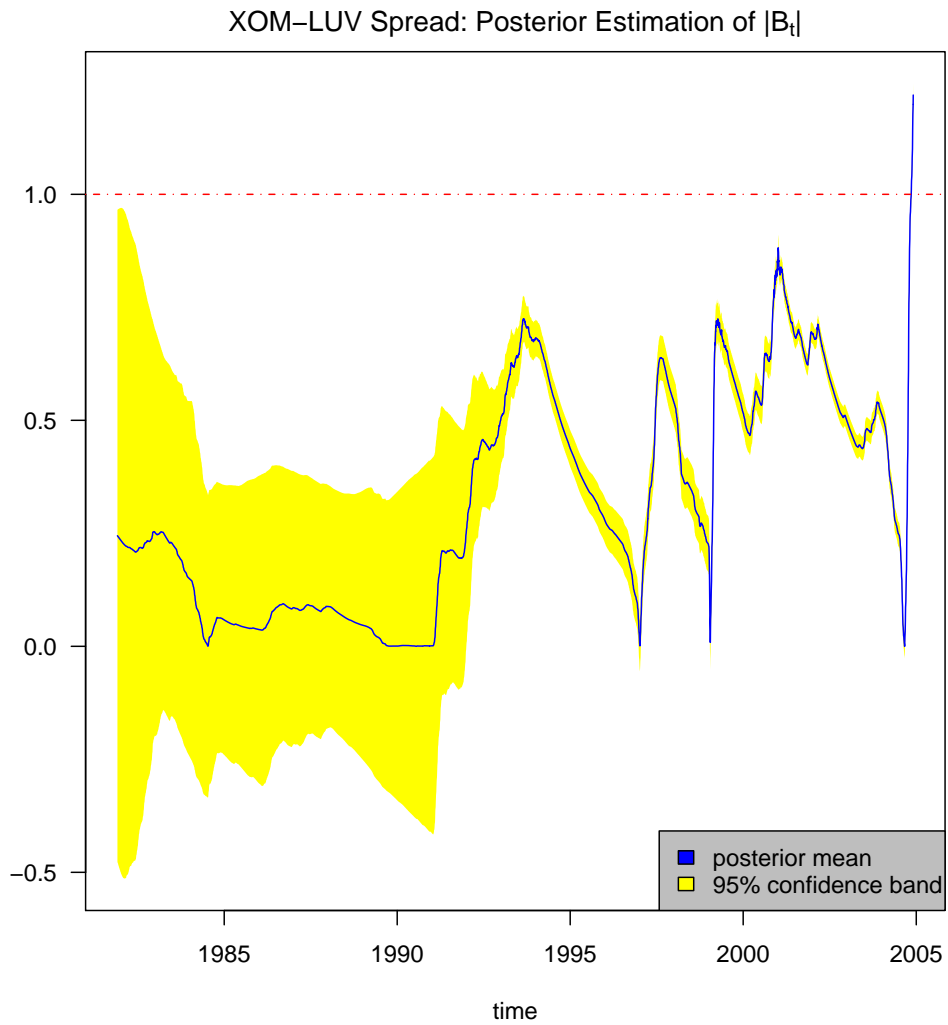}
\caption{Posterior estimation of $|B_t|$. We have used $\phi_1=0.1$,
$\phi_2=99839$, $\delta_1=0.992$, $\delta_2=0.995$ and a prior
$P_1=1000I_2$.} \label{spread1B}
\end{figure}

\newpage

\begin{figure}
\includegraphics[scale=0.75]{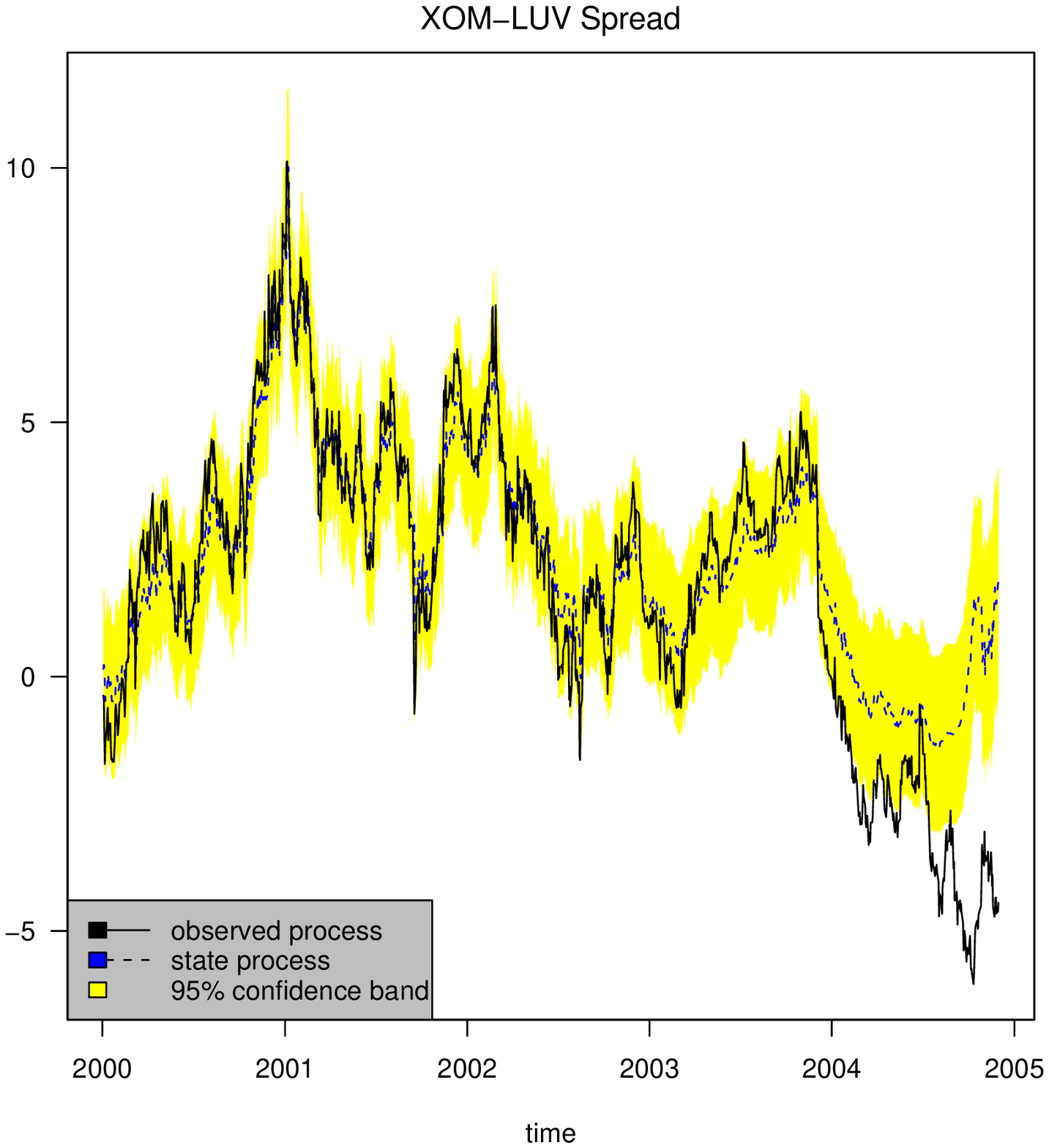}
\caption{Observed spread and state spread using a recursive
regression routine for on-line spread availability.}
\label{spread1}
\end{figure}

\newpage

\begin{figure}
\includegraphics[scale=0.75]{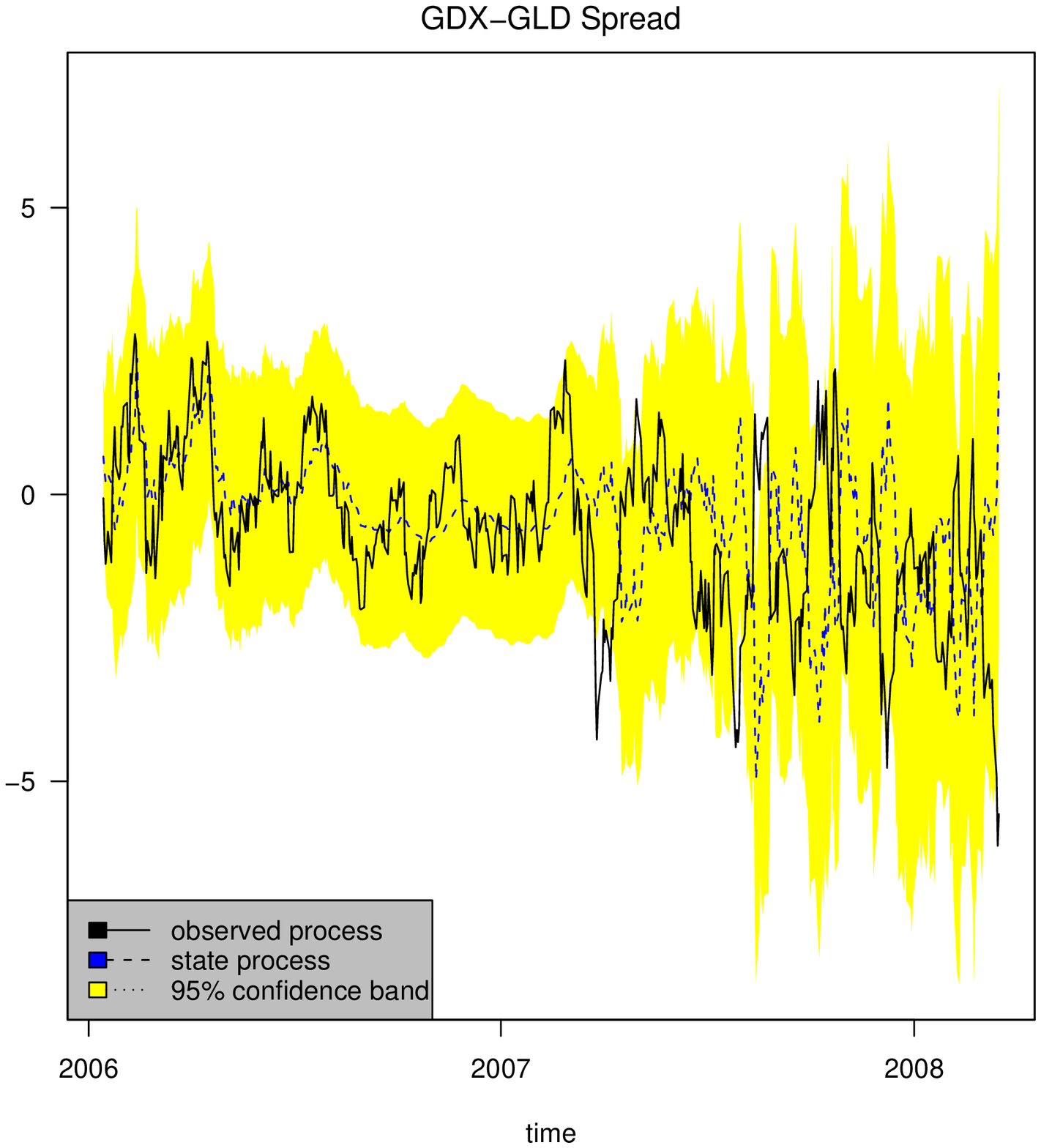}
\caption{Observed spread and state spread using a recursive
regression routine for on-line spread availability. }
\label{spread2}
\end{figure}

\newpage
\begin{figure}
\includegraphics[scale=0.75]{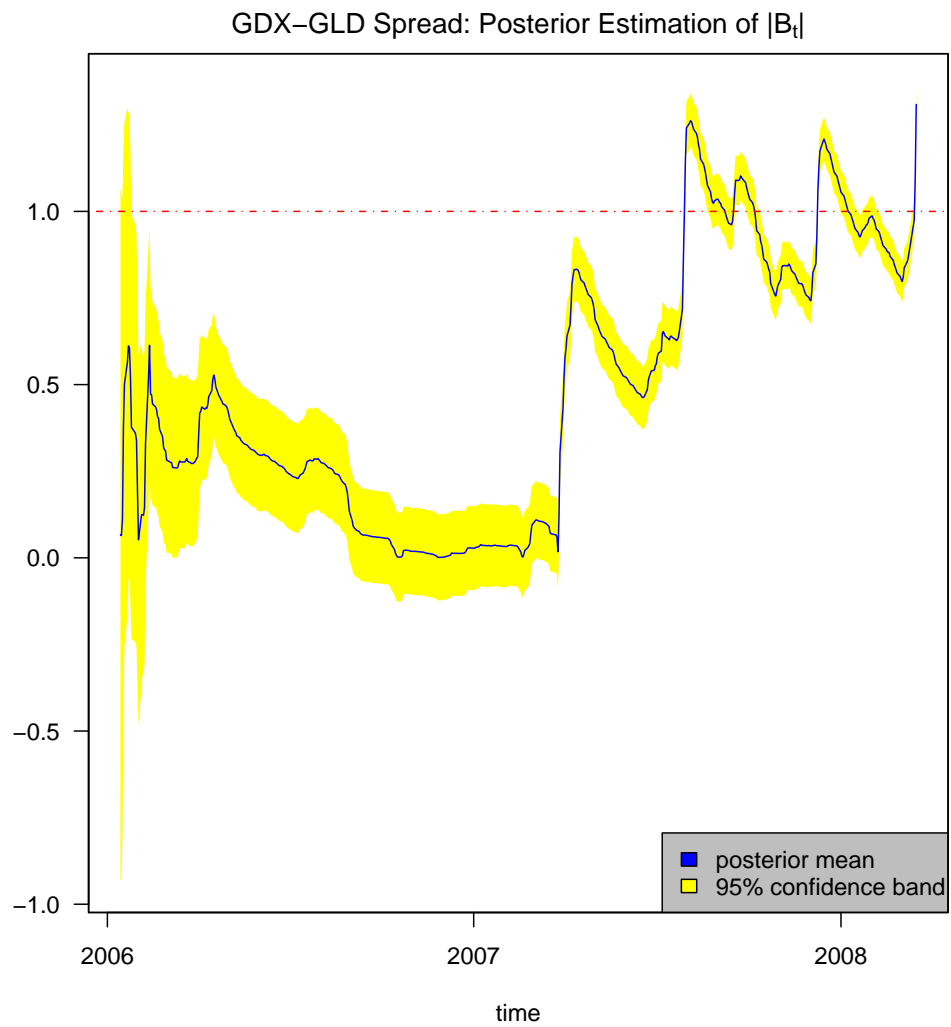}
\caption{Posterior estimation of $|B_t|$. We have used
$\phi_1=0.999$, $\phi_2=99$, $\delta_1=0.95$, $\delta_2=0.98$ and a
prior $P_1=1000I_2$.} \label{spread2B}
\end{figure}

\end{document}